\documentclass[a4paper,12pt]{article}
\usepackage{graphicx}
\usepackage{amsmath,amsthm,amsfonts,amssymb,mathrsfs,mathtools,bm,soul}
\usepackage[OT1]{fontenc} 
\usepackage{etoc}
\usepackage{float}
\usepackage{setspace}
\usepackage{fancyhdr} % clear all header and footer field

\def\subsectiontitle{}
\def\subsubsectiontitle{}
\usepackage{tikz}
\usepackage[colorinlistoftodos,textsize=scriptsize]{todonotes}

\usepackage{inconsolata}

\usepackage{enumerate}
\usepackage[margin=1in,marginparwidth=1in]{geometry}
\usepackage[colorlinks,breaklinks=true]{hyperref}
\usepackage{breakcites}
\usepackage{xcolor}
\AtBeginDocument{%
	\hypersetup{%
    linkcolor={red!50!black},%
    citecolor={blue!50!black},%
    urlcolor={blue!80!black}%
}%
}%

\usepackage[nameinlink,noabbrev,sort,capitalise]{cleveref}
\makeatletter
\def\ps@pprintTitle{%
 \let\@oddhead\@empty
 \let\@evenhead\@empty
 \def\@oddfoot{\emph{Very preliminary version}\hfill\emph{This draft: \today}}%
 \let\@evenfoot\@oddfoot}
\makeatother
\usepackage{etoolbox}
\patchcmd{\pprintMaketitle}
  {\fi\hrule}% the second rule
  {\fi\ifvoid\extrainfobox\else\unvbox\extrainfobox\par\vskip10pt\fi\hrule}
  {}{}

\newsavebox\extrainfobox

\newtheorem{prop}{Proposition}
\crefname{prop}{Proposition}{Propositions}

\newtheorem{thm}{Theorem}
\crefname{thm}{Theorem}{Theorems}

\crefname{cor}{Corollary}{Corollaries}

\newtheorem{lem}{Lemma}
\crefname{lem}{Lemma}{Lemmas}

\crefname{ass}{Assumption}{Assumptions}

\crefname{axiom}{Axiom}{Axioms}

\newtheorem{defi}{Definition}
\crefname{defi}{Definition}{Definitions}

\theoremstyle{remark}

\theoremstyle{definition}

\crefname{eg}{Example}{Examples}

\crefname{problem}{Problem}{Problems}

\newcommand{\overbar}[1]{\mkern 1.5mu\overline{\mkern-1.5mu#1\mkern-1.5mu}\mkern 1.5mu}

\usepackage{indentfirst}

\allowdisplaybreaks

\let\oldfootnote\footnote
\renewcommand\footnote[1]{\oldfootnote{\hspace{.4mm}#1}}
\makeatletter
\renewenvironment{proof}[1][\proofname] {\par\pushQED{\qed}\normalfont\topsep6\p@\@plus6\p@\relax\trivlist\item[\hskip\labelsep\bfseries#1\@addpunct{.}]\ignorespaces}{\popQED\endtrivlist\@endpefalse}
\makeatother
\let\oldFootnote\footnote
\newcommand\nextToken\relax

\renewcommand\footnote[1]{%
    \oldFootnote{#1}\futurelet\nextToken\isFootnote}

\newcommand\isFootnote{%
    \ifx\footnote\nextToken\textsuperscript{,}\fi}
\usepackage{natbib}
\def\d{\mathrm{d}}
\def\E{\mathbb{E}}
\def\M{\mathcal{M}}
\def\F{\mathbb{F}}
\def\G{\mathbb{G}}
\def\R{\mathbb{R}}

\def\k{k^{th}}
\def\n{\circledR}

\usepackage{todonotes}

\begin{document}

\title{\fontsize{15}{18} \selectfont \textbf{Rank-Guaranteed Auctions}\footnote{We thank Mohammad Akbarpour, Martin Bichler, Tilman B\"orgers, Kim-Sau Chung, Rahul Deb, Jason Hartline, Paul Klemplerer, Fuhito Kojima, Andrew Komo, Paul Milgrom, Michael Ostrovsky, Antonio Penta, Ning Sun, Satoru Takahashi, and Andrzej Skrzypacz for helpful discussions. This manuscript subsumes \cite{he2022order}.}}

\author{
Wei He\thanks{Department of Economics, The Chinese University of Hong Kong, hewei@cuhk.edu.hk}
\and Jiangtao Li\thanks{School of Economics, Singapore Management University, jtli@smu.edu.sg}
\and Weijie Zhong\thanks{Graduate School of Business, Stanford University, weijie.zhong@stanford.edu}
}

\maketitle

\begin{abstract}
	
We propose a combinatorial ascending auction that is ``approximately'' optimal, requiring minimal rationality to achieve this level of optimality, and is robust to strategic and distributional uncertainties. Specifically, the auction is \emph{rank-guaranteed}, meaning that for any menu $\M$ and any valuation profile, the ex-post revenue is guaranteed to be at least as high as the highest revenue achievable from feasible allocations, taking the \((|\M| + 1)^{th}\)-highest valuation for each bundle as the price. Our analysis highlights a crucial aspect of combinatorial auction design, namely, the design of menus. We provide simple and approximately optimal menus in various settings.

%Jiangtao: I rewrote the abstract to highlight the important aspect of menu design, and removed the sentence on robust auction design. Even under distributional uncertainty of valuations, the rank-guarantee remains asymptotically robustly optimal, with the difference from the worst-case total surplus being at most $O(\frac{1}{N})$.

\end{abstract}

\newpage

\section{Introduction}

Auctions play a critical role in economic activities. For example, the online advertising sector generates trillions of dollars annually through the auctioning of advertising ``slots''. The Federal Communications Commission (FCC) has collected over 200 billion dollars via auctioning radio spectrum. Despite the critical role these auctions play, there is a surprising lack of theoretical groundwork to navigate the intricacies of auction design. This gap in knowledge stems from a unique challenge: ``\emph{not all slots are created equal}''---bidders typically have complex, combinatorial preferences for different slots. For instance, YouTube intersperses promotional videos at regular intervals within longer content. Here, some advertisers might see value in the repetition of their ads, leveraging the complementarity, while others may fear overexposure could lead to negative perceptions akin to ``spamming.'' Likewise, in the realm of telecommunications, while larger service providers may pursue nationwide radio spectrum licenses to maximize their coverage, smaller providers often seek only regional licenses, prioritizing local markets over national presence.

The theoretical study of auctions, in contrast to reality, takes aggressive simplifications to an extent that overlooks the nuanced realities of the markets. The game theoretic approach focuses on the incentives of the bidders (incentive compatibility) and auctioneers (optimality), while limiting to very specific environments. Iconic theories, such as \cite{my81}, made progress by assuming a \emph{single item}, \emph{independent valuations} among bidders, full \emph{Bayesian rationality}, and \emph{shared prior beliefs}. None of these assumptions hold water in the complex scenarios described earlier, highlighting a clear disconnect. Conversely, the domain of operations research and computer science typically emphasizes the procedural aspects of communicating preferences and determining assignments under complete preferential complexity. This focus often neglects the incentives for bidders to reveal their preferences truthfully and for auctioneers to maximize revenue (see, e.g., a review by \cite{cramton2006combinatorial}). Therefore, the development of a comprehensive auction theory that simultaneously addresses both \emph{incentives} and \emph{complexity} represents a significant and unfulfilled challenge.

The goal of this study is to address the dilemma between incentives and complexity in auction design. This resonates with what \cite{carroll2019design} speculates as the ``future of economic design'':
\begin{quote}
    ``My expectation---and my hope---is that progress in mechanism design over the coming decades will come from developing more useful general models of preferences, information, and actions in complex environments, relatively free of structural assumptions; and developing conceptual tools to argue for why certain kinds of mechanisms will work well in such environments.''
\end{quote}
To tackle the complexity of the practical environments, the same article then proposes:
\begin{quote}
    ``Making this progress toward more free-form models will require \ul{a shift in the criteria} by which research in economic theory is evaluated.''
\end{quote}
The methodology we develop in this paper hinges on the shift to a novel criterion of approximate optimality, termed the \textbf{rank guarantee}: the $\k$-rank guarantee is the maximal ex-post revenue when each feasible bundle can be sold at the $\k$-highest value among all bidders. 
%This reality starkly contrasts with traditional auction theory, which often simplifies to an extent that overlooks the nuanced realities of the markets. Iconic theories, such as \cite{my81}, typically focus on the allocation of a \emph{single item}, assuming \emph{independent valuations} among participants, along with both the auctioneer and bidders possessing fully \emph{Bayesian rationality}, underpinned by accurate, \emph{shared prior beliefs}. However, none of these assumptions hold water in the complex scenarios described earlier, highlighting a clear disconnect. The development of a comprehensive theory for maximizing revenue through auction design in these more intricate and generalized environments remains a significant, unmet challenge.
With this approximation, we achieve a near-optimal resolution of the dilemma: it is possible to design a simple auction mechanism to respect the incentives---the auctioneer achieves a rank-guaranteed revenue as long as bidders avoid ``obviously'' bad strategies---while accommodating the full complexity---fully combinatorial preferences, without the need for any Bayesian prior. 

In our model, an auctioneer sells multiple items to several (potentially a large number of) strategic bidders, each with private valuations. We introduce a multi-item variant of the open ascending auction termed the (C)ombinatorial (As)cending (A)uction (CASA). Prior to the auction, the auctioneer curates a \emph{menu} of item bundles for allocation, denoted by $\M$. With the formal game theoretic form of the auction described in \cref{sec:CASA}, this auction model distills down to two straightforward principles:

\begin{enumerate}
\item Bidders are allowed to place binding bids (increase prices) on any assortment of bundles from the menu, even if these selections overlap.
\item The auction concludes when bid prices stabilize, with the winning bids being those that maximize the total selling price.
\end{enumerate}

Our findings reveal that CASA resolves the dilemma within certain approximation bounds. First, we show that the outcome of CASA respects the incentives of both bidders and the auctioneer:
\begin{itemize}
    \item \textbf{Bidder rationality}: All of our results apply to any \emph{non-obviously dominated} strategy profile, a collection that excludes strategies that are \emph{obviously dominated} in the sense that even in the most favorable case, they underperform some other strategies in their least favorable case (see \cite{li2017obviously} and \cite{li2021simple}). In other words, we allow the bidders to bid fully strategically, while making a weak rationality assumption that all we know is that they avoid obviously bad choices.
    
    \item \textbf{Approximate auctioneer optimality}: In CASA, any non-obviously dominated strategy profile yields an \emph{ex-post} revenue that is \emph{rank-guaranteed} --- achieving the maximal revenue when each bundle within the menu $\M$ can be sold at the $(|\M|+1)^{th}$-highest value among all bidders. Since the $1^{st}$-guarantee is the full ex-post trading surplus, the approximation we take is close in ``ordinal distance'' from the full surplus.
\end{itemize}
Second, we show how CASA maintains its performance in complex, unknown environments:
\begin{itemize}
    \item \textbf{Prior-free}: The auction format, the selection of strategy profile, and the revenue guarantee does not depend on any Bayesian prior on either the auctioneer's side or the bidders' side, rendering the mechanism \emph{prior-free} and our rank guarantee a prior-free approximation (see, e.g., Chapter 5 and 7 of \cite{hartline2013mechanism}).
    \item \textbf{Distributional robustness}: We quantify the rank-guarantee using canonical robust optimality criteria, i.e., the minimal ex-ante expectation when an adversarial nature chooses the joint distribution of values against the mechanism (see, e.g., \cite{gc17}). In the worst case, the revenue from CASA approximates the total surplus at the rate of $O\left(\frac{|\M|^2}{N}\right)$, i.e., CASA asymptotically achieves full surplus extraction when the number of bidders is large relative to the menu size.
\end{itemize}

To our knowledge, this paper is the first to introduce the rank-guarantee approximation. Compared to canonical approximation notions like constant-ratio guarantee and maxmin guarantee, the rank guarantee offers the added benefit of being easy to evaluate across various Bayesian and non-Bayesian environments. On the one hand, the rank guarantee provides an easily computable lower bound on the worst-case revenue, even when the underlying environment changes, such as when there are more bidders, when the menu changes, or when the auctioneer has additional information on the distribution of bidders' valuations. Such adaptivity allows us to further simplify CASA by studying menu design. On the other hand, the lower bound performance is also straightforward to assess outside adversarial scenarios. For example, in the canonical setting where the values are independent and identically distributed, the rank guarantee is an appealing approximation when \(N\) is large, as all order statistics converge to the upper bound of the valuation support. More generally, beyond understanding the worst-case scenario, the rank guarantee can be useful to consider the potential upside of a mechanism in best-case scenarios.

Our framework highlights a crucial aspect of combinatorial auction design, namely, the design of menus. Crucially, the rank guarantee we derive reveals a novel trade-off between \emph{menu sufficiency} and \emph{approximation efficiency}: a more complete menu achieves a higher benchmark total surplus but increases the rank $|\M|$. Therefore, to close the approximation gap, a key exercise is to reduce the menu size while maintaining the allocation efficiency, leveraging further knowledge about the bidders' preferences. We focus on a specific type of \emph{sufficient} menus that improve approximation efficiency ``for free''---menus that achieve the same benchmark total surplus as the complete menu. Specifically, we show that when the bidder's preference exhibits canonical preference structures, without loss of the benchmark total surplus, the size of menus can be reduced to be polynomial in the number of items being auctioned and so is the convergence rate of revenue guarantee. The result is summarized in \cref{tab:subfficient}. \medskip

\begin{table}[htbp]
\centering
\begin{tabular}{|c|c|c|}
\hline
\textbf{Preference} & \textbf{Simple and Sufficient Menu}  &  \textbf{Rank $k$} \\
\hline\hline
Weak substitutability & Individual items & $O(M)$ \\
\hline
Weak complementarity & Grand bundle  & $2$ \\
\hline
Partitional complementarity & Partitional bundles & $O(M)$ \\
%{\footnotesize(with $I$ different possible partitions)} & & {\footnotesize $O(M\times I)$}\\
\hline
Homogeneous goods & Menu of quantities  & $O(M^2)$ \\
%{\footnotesize(with $I$ heterogeneous types)}& & {\footnotesize $O(M^I)$}\\
\hline
\end{tabular}
\caption{Simple and sufficient menus}
\label{tab:subfficient}
\end{table}

The remainder of the introduction reviews related literature. \Cref{sec:CASA} introduces the auction format of CASA and the notion of rank-guarantee, and shows that CASA achieves the rank-guarantee. \Cref{sec:maxmin} bounds the worst-case performance of rank-guaranteed auctions under distributional uncertainties. \Cref{sec:sufficient} explores specific preference structures where CASA with simple menus performs as well as the complete menu.

\subsection{Related literature}

\paragraph{(Approximately) optimal auction design} Beyond the simple environment studied in \cite{my81} and \cite{bulow1996auctions}, solving for the exact optimal mechanism with confounding factors like multiple heterogeneous items, bounded distributional knowledge or bounded rationality is generally intractable. Various alternative optimality notions have been proposed to make progress (see surveys by \cite{roughgarden2015approximately} and \cite{hartline2013mechanism}). \cite{aggarwal2006knapsack} and \cite{goldberg2001competitive} obtained the ``constant fraction'' approximation in the auction of sponsored search and digital goods. Following a broader literature on robust mechanism design pioneered by \cite{gc17}, various authors have studied ``robustly'' optimal auctions that maximize the distributional worst-case revenue.\footnote{This research direction complements the large body of papers that focus on the case in which the designer does not have reliable information about the agents' hierarchies of beliefs about each other while assuming the knowledge of the payoff environment; see, for example, \cite{bm05}, \cite{ce07}, \cite{cl18}, \cite{du18}, \cite{bd19}, \cite{yamashita2022foundations}, and \cite{brooks2023structure}.} Particularly, \cite{hl21}, \cite{zhang2022correlationrobust}, and \cite{suzdaltsev2022distributionally} study robust versions of the single-unit auction problem in the distribution robust framework where the auctioneer has non-Bayesian uncertainty about the joint distribution of the bidders' valuations. In this paper, we propose the notion of rank-guarantee, a distinct notion of approximate optimality. In \cref{sec:maxmin}, we apply the distributional robustness analysis similar to that of \cite{gc17} and show that the rank-guarantee has an appealing worst-case performance.  

\paragraph{Multi-item auctions} Beyond the efficient Vickrey auction, few theoretic results have been established regrading multi-item auctions with combinatorial preferences. \cite{jehiel2001efficient} point out the vulnerability of efficiency under multidimensional bidder information. \cite{ausubel2002ascending} point out the poor revenue performance and strategic vulnerability of the Vickrey auction and propose simultaneous ascending auctions with package bidding (SAAPB). The multi-item auction design problem has also been extensively studied in the field of combinatorial auctions (see \cite{cramton2006combinatorial} for a survey). This literature mainly focuses on (approximately) efficient auction design and their communication/computational complexity, which is orthogonal to our focus on revenue performance and bidder incentives. Compared to other proposals like SAAPB (\cite{ausubel2002ascending}) and the Combinatorial Clock Auction (CCA, see \cite{ausubel2006clock} and \cite{levin2016properties}), CASA uses a simpler ``pay-as-bid'' rule, as opposed to personalized prices in SAAPB and demand reporting in CCA. Importantly, we assume that the bidders are fully strategic instead of single-minded. The menu design problem we tackle in \cref{sec:sufficient} is akin to \cite{rothkopf1998computationally}, which seeks to make a combinatorial auction computationally tractable by restricting the menu. We achieved the same goal while allowing for strategic bidders and maintaining the approximate optimality of the revenue.

%Jiangtao: I have removed the following discussion. Our result overturned the seminal impossibility theorem of \cite{roughgarden2014barriers} due to a different notion of simplicity: in the dynamic CASA, the \emph{strategic form} game violates the simplicity requirement of \cite{roughgarden2014barriers}; meanwhile, it is straightforward to identify obviously optimal \emph{behavioral} strategies in the \emph{extensive form} game.

\paragraph{Implementation in strategies that are not obviously dominated} We study outcomes when agents are rational in the sense of avoiding obviously dominated strategies. This solution concept draws from the idea of obvious strategy-proof mechanisms (see \cite{li2017obviously}) and is systematically studied in \cite{li2021simple}. As in \cite{li2021simple}, we assume that agents avoid obviously dominated strategies, but refrain from making assumptions regarding how agents select among strategies that are not obviously dominated. This methodology aligns with the spirit of implementation in undominated strategies; see for example \cite{carroll2014complexity}, \cite{borgers1991undominated}, \cite{jackson1992implementation}, and \cite{yamashita2015implementation}.

\section{CASA and rank-guarantee}

\label{sec:CASA}

\subsection{The auction environment}

\label{sec:environment}

There is a set $S$ of $M$ items to be sold to $N$ bidders. Let \(\mathcal{N} = \{1, 2, \ldots, N\}\). We write \( b \subseteq S \) to denote a generic bundle of items. Let \( \bm{v}^n =  \{v_b^n\}_{b \subseteq S} \) denote the valuation vector of bidder \(n\), where $v_b^n$ is bidder $n$'s valuation of bundle $b$. Valuations are normalized so that \( v^n_\emptyset = 0 \) and  \( v_b^n\in [\underline{v}, \overbar{v}] \) (\( \underline{v} \ge 0 \)) for all $b \neq \emptyset$. A generic valuation profile is denoted by \( \bm{v} = (\bm{v}^1, \bm{v}^2, \ldots, \bm{v}^N)\). Let $\M \subseteq 2^S$ denote a \emph{menu} of bundles chosen by the auctioneer. While $\M$ is a choice variable of the auctioneer, for now we take it as exogenously given; we defer the discussion of menu design to \cref{sec:sufficient}. Assume that \( N \ge |\M| + 1 \). Let
\begin{align*}
  \mathcal{B}(\M) = \{ X \subseteq \M \, |  \, \forall b, b' \in X, \, b \cap b' = \emptyset \}
\end{align*}
denote the set of \emph{feasible} allocations of bundles within the menu $\M$, i.e., all collections consisting of non-overlapping bundles. 

\subsection{The Combinatorial Ascending Auction}

\label{sec:auction-format}

We define the \emph{Combinatorial Ascending Auction} (CASA) as follows. The auction has an iterative structure, with the ``state of the auction'' characterized by the identity of the leading bidder and the leading price for each bundle. Initially, the leading price for each bundle is zero and none of the bidders is a leading bidder for any bundle. Bidders take turns raising the bids on the bundles, which determines new leading bidders and leading prices. The process repeats itself until when there are no new bids on any bundle. At that point, the auction stops. The auctioneer chooses a feasible allocation to maximize revenue, taking the leading prices as the prices for the bundles. There is also an activity rule designed to ensure that bidding activity starts out high and declines during the auction as prices rise far enough to discourage some bidders from continuing.

Formally, let $P \subset \R^+$ be a finite grid of feasible bids with grid size $\epsilon$ and $\max P > \overbar{v}$.
\begin{enumerate}[(1)]
\item \textbf{Initialization stage $t=0$}. Define
\begin{itemize}
\item the \emph{leading bidder vector} at stage $0$: $\bm{\phi}^0=(\phi^0_b)_{b\in \M}= \bm{0}$,
\item the \emph{leading price vector} at stage $0$: $\bm{p}^0=(p^0_b)_{b\in \M}=\bm{0}$,
\item the set of \emph{active bidders} at stage $0$: \(\mathcal{N}^0 = \{1, 2, \ldots,N\}\).
\end{itemize}
\item \textbf{Bidding stage $t \ge 1$}. An active bidder \(n\in\mathcal{N}^{t-1}\) observes $(\bm{p}^{t-1}, \{b \, | \, \phi_b^{t-1}=n\})$, and decides whether to quit, which bundles to bid on, and how much to bid.\footnote{The bidder selection rule and the observability of history is inconsequential for our analysis. For concreteness, we consider the selection rule that active bidders are cycled in ascending order according to their indices, and the observability of history is minimized to maximally protect privacy.}
\begin{itemize}
\item Bidder $n$ may choose to quit by submitting an empty bid only if $\{b \, | \, \phi_b^{t-1}=n\}=\emptyset$, i.e., bidder $n$ is not a leading bidder for any bundle in stage $t - 1$. Quitting is irreversible, that is, if bidder $n$ chooses to quit, then bidder $n$ becomes an inactive bidder and does not participate in future bidding rounds. Update:
\begin{itemize}
\item $\bm{\phi}^t=\bm{\phi}^{t-1}$, $\bm{p}^t=\bm{p}^{t-1}$, $\mathcal{N}^t = \mathcal{N}^{-1} \setminus \{n\}$.
\end{itemize}
\item If bidder $n$ chooses not to quit, then she submits a bid---a nonempty set of bundle-price pairs \(\{(b,p_b)\}\subset \M \times P\), subject to the requirements that (1) Leading bids are binding: if bidder $n$ is the leading bidder at some bundle in stage $t - 1$, then she must include that bundle in her bid with a bid that is \textit{weakly} higher than the current leading price for that bundle, and (2) Minimum bid increment: if bidder $n$ would like to bid on some bundle for which she is not the leading bidder, then her bid for that bundle must be \textit{strictly} higher than the current leading price for that bundle. Update: 
\begin{itemize}
\item $\phi_b^t = n$ and \(p_b^t = p_b\) for any bundle $b$ included in her bid,
\item $\phi_{b'}^t = \phi_{b'}^{t-1}$ and $p_{b'}^t=p_{b'}^{t-1}$ for any bundle $b'$ not included in her bid,
\item $\mathcal{N}^t=\mathcal{N}^{t-1}$.
\end{itemize}
\end{itemize}
Then, move on to the bidding stage $t+1$. 
\item \textbf{Allocation}. The auction ends (in stage $T$) when the leading prices stay constant for $N$ consecutive periods. The auctioneer chooses a feasible allocation to maximize
\begin{align*}
\max_{\bm{b}\in \mathcal{B}(\M)} \, \sum_{b\in \bm{b}} \, p_b^T.
\end{align*}
Denote the maximizer by $\bm{b}^*$. Each bundle $b\in\bm{b}^*$ is allocated to $\phi_b^T$ at the price $p_b^T$.
\end{enumerate}

%Jiangtao: For the auction to end, I change the requirment of constant prices for N + 1 periods to N periods.

In words, the auction format of CASA runs parallel ascending auctions for each bundle \(b \in \M\). Then, the items are allocated to maximize the total price. We discuss possible variants of CASA and its relation to existing auction formats in \cref{sec:discussion}.

\subsection{The rank-guarantee of CASA}

\label{sec:rank-guarantee}

In this subsection, we study the strategic behavior of the bidders and establish the rank-guarantee property of CASA. We only assume minimal rationality on the part of the bidders---bidders are rational in the sense of not playing obviously dominated strategies. Formally we adopt the solution concept of implementation in strategies that are not obviously dominated (see \cite{li2017obviously} and \cite{li2021simple}). We first sketch the intuition, and then provide the formal arguments.

At any history, consider a non-leading bidder's choice as to whether to quit. Obviously, as quitting is irreversible, quitting the auction leads to a best possible outcome of a zero payoff. Suppose that there is some bundle for which the bidder's valuation is higher than the current leading price for that bundle. Consider the following strategy where the bidder raises the price for this particular bundle and never revises her bid afterwards. Clearly, this continuing strategy guarantees a non-negative payoff for the bidder. Thus, at least for the purpose of deciding whether to quit, it is ``obviously optimal'' not to quit.

More formally, let $h = (t, (\mathcal{N}^0, \ldots, \mathcal{N}^{t-1}), (\bm{p}^0\ldots,\bm{p}^{t-1}), (\bm{\phi}^0,\ldots,\bm{\phi}^{t-1}))$ denote a history of the game in stage $t$, $H_t$ the set of such histories in stage $t$, and $H = \cup_{t \ge 0} H_{t}$. 
%Let $I_n(\bm{p},\bm{b})$ be bidder $n$'s information set given observed prices $\bm{p}$ and $n$'s leading bundles $\bm{b}$ when $n$ is the active bidder. Then for any $h\in I_n(\bm{p},\bm{b})$, $\bm{p}^{t-1}=\bm{p}$ and $\{b|\phi_b^{t-1}=n\}=\bm{b}$ for some $t$. 
Suppose that bidder $n$ is the active bidder in some stage $t$ and $I_n$ is bidder $n$'s information set. Then the observed prices $\bm{p}$ and $n$'s leading bundles $\bm{b}$ are the same for all $h \in I_n$. Let $\mathcal{I}_n$ denote all information sets of $n$. Let $s_n:\mathcal{I}_n\to 2^{\M\times \R^+}$ denote bidder $n$'s (pure behavioral) strategy and $u_n(\bm{s},\bm{v}^n|h)$ the payoff to bidder $n$ given valuation vector $\bm{v}^n$, strategy profile $\bm{s}$, conditional on the current history $h$ and $n$ bidding in period $t$.

\begin{defi}
    A bidding strategy $s_n: \mathcal{I}_n \to 2^{\M\times P}$ is \emph{obviously dominated} if there exists $s_n'$ such that at any earliest point of departure $I_n$ between $s_n$ and $s_n'$,
    \begin{align*}
        \sup_{s_{-n},h\in I_n}u_n(\bm{s},\bm{v}^n|h)&\le \inf_{s_{-n},h\in I_n}u_n(s'_n,\bm{s}_{-n},\bm{v}^n|h);\\
        \inf_{s_{-n},h\in I_n}u_n(\bm{s},\bm{v}^n|h)&< \sup_{s_{-n},h\in I_n}u_n(s'_n,\bm{s}_{-n},\bm{v}^n|h).
    \end{align*}
\end{defi}
The first inequality is identical to the definition of the obvious dominance relation in \cite{li2017obviously}, i.e., the best outcome under $s_n$ is weakly worse than the worst outcome under $s_n'$. In addition, we require the dominated strategy to be non-equivalent in terms of the induced outcome to the strategy that dominates it. The second requirement guarantees that the set of non-obviously dominated strategies is non-empty. The earlier intuition then translates to:

\begin{lem}\label{lem:order}

If there exists an information set $I_n\in\mathcal{I}_n$ (with observed prices $\bm{p}$) such that $s_n(I_n)=\emptyset$ (i.e., bidder $n$ quits) and 
\begin{itemize}
\item $\exists \bm{p}' \in P$, $\tilde{\bm{b}} \in \arg\max\limits_{\bm{b}\in\mathcal{B}(\M)} \, \sum_{b'\in\bm{b}} \, p'_{b'}$ and $\displaystyle b \in \tilde{\bm{b}}$ such that $\bm{p}'\ge \bm{p}$ and $v^n_b> p'_b> p_b$,
\end{itemize}
then $s_n$ is obviously dominated.

\end{lem}
\begin{proof}

To show that $s_n$ is obviously dominated, we explicitly construct another strategy $s_n'$ that obviously dominates it. Consider the information set $I_n \in \mathcal{I}_n$ (with observed prices $\bm{p}$) such that $s_n(I_n)=\emptyset$ and 
\begin{itemize}
\item $\exists \bm{p}' \in P$, $\tilde{\bm{b}} \in \arg\max\limits_{\bm{b}\in\mathcal{B}(\M)} \, \sum_{b'\in\bm{b}} \, p'_{b'}$ and $\displaystyle b \in \tilde{\bm{b}}$ such that $\bm{p}'\ge \bm{p}$ and $v^n_b> p'_b> p_b$.
\end{itemize}
Obviously, the payoff from $s_n$ conditional on any history $h\in I_n$ is zero. Let $s_n'$ be the same as $s_n$ before the information set $I_n$ and 
let bidder $n$ bid $s_n'(I_n)=(b,p'_b)$ and never revise her bid afterwards.

We first discuss a special case in which $|\mathcal{N}^{t-1}|=1$ and $n$ is not a current leading bidder (otherwise quitting the auction is not feasible for bidder $n$). Then, all the current prices must be $0$ and all the other bidders have quit (as this is the unique consistent history). Following the strategy $s_n'$, the auction ends with bidder $n$ bidding $(b,p'_b)$, leading to a positive payoff of $v_b^n - p_b' >0$ for bidder $n$.

Next, we consider the case in which $|\mathcal{N}^{t-1}|>1$.  It is clear that bidder $n$ will get a nonnegative payoff regardless of the value profiles and bidding strategies of other bidders. We show that the best possible payoff for bidder $n$ following the strategy $s_n'$ is positive. It suffices to consider the case in which any other subsequent active bidder bids $p'_{b'}$ for some bundle $b'$ whenever possible, the auction ends with prices $\bm{p}'$, leading to a positive payoff of $v_b^n - p_b' > 0$ for bidder $n$.
\end{proof}
Note that \cref{lem:order} suggests that our intuition is incomplete as we cannot fully rule out strategies that quit when the bidder's value for some bundle is above the leading price for that bundle. To rule out all such strategies, it further requires the existence of some scenario under which bidder $n$ may be pivotal: a strictly profitable bid of $n$ may ever be selected in the end. Nevertheless, we show below that the non-pivotal bidders are inconsequential for our analysis. 

Let $S^n_{NOD}(P)$ denote the set of non-obviously dominated strategies of $n$ given price grid $P$. Let $R(\bm{s},\bm{v})$ denote the revenue to seller given value profile $\bm{v}$ and strategy profile $\bm{s}$. Define
\begin{align*}
    \underline{R}_{CASA}(\bm{v}):=\varliminf_{\epsilon\to 0} \, \inf_{s_n\in S^n_{NOD}(P)} \, R(\bm{s},\bm{v}),
\end{align*}
where the first limit inf is taken over $\epsilon\to 0$ and all $P$ with grid size $\epsilon$. That is, \(\underline{R}_{CASA}(\bm{v})\) is the worst-case ex-post revenue from CASA under non-obviously dominated strategies in the limit where grid $P$ becomes dense.

Given the valuation profile \(\bm{v}\) and menu \(\M\), the \(\k\)-guarantee is defined as the maximal revenue from feasible allocations within \(\M\), taking the \(\k\)-highest valuations for each bundle as the price.

\begin{defi}

The \(\k\)-guarantee given the menu \(\M\) and the value profile $\bm{v}$ is
$$\displaystyle R^k_{\M}(\bm{v}):=\max_{\bm{b} \in \mathcal{B}(\M)} \, \sum_{b \in \bm{b}} \, v^{(k)}_b,$$
where $v^{(k)}_b$ denotes the $\k$-highest value of bundle $b$.

\end{defi}
Our key observation is that CASA achieves the \(\k\)-guarantee as long as bidders avoid strategies that are obviously dominated. In other words, for CASA to achieve the \(\k\)-guarantee, we only need minimal rationality on the part of the bidders.

\begin{thm} \label{thm:1}

\(\displaystyle\underline{R}_{CASA}(\bm{v})\ge{R}^k_{\M}(\bm{v})\) for $k=|\M|+1$.

\end{thm}

\begin{proof}

Consider a grid $P$ with any grid size $\epsilon > 0$. We show that as long as bidders avoid obviously dominated strategies, for any value profile $\bm{v}$,
\begin{align*}
\max_{\bm{b}\in \mathcal{B}(\M)} \, \sum_{b\in \bm{b}} \, p^T_b \ge \max_{\bm{b}\in \mathcal{B}(\M)} \, \sum_{b\in \bm{b}} \, (v^{(k)}_b-\epsilon).
\end{align*}
Suppose to the contrary, this is not true. Then, there exists strategy profile $s$ (with $s_n$ being not obviously dominated for each bidder $n$), value profile $\bm{v}$, and $\delta > 0$ such that \begin{align*}
\max_{\bm{b} \in \mathcal{B}(\M)} \, \sum_{b\in \bm{b}} \, p^T_b < \max_{\bm{b} \in \mathcal{B}(\M)} \, \sum_{b\in \bm{b}} \, (v^{(k)}_b -\epsilon - \delta).
\end{align*}
Evidently, there exists some bundle $b$ such that $p_b^T < v_b^{(k)}-\epsilon-\delta$, or equivalently, $p_b^T + \epsilon < v_b^{(k)} - \delta$. For each such bundle $b$, raise the price of bundle $b$ to (the closest price below) $v_b^{(k)} - \delta$ sequentially until we find the first pivotal bundle $\tilde{b}$ when $\max_{\bm{b}\in \mathcal{B}(\M)}\sum_{b\in \bm{b}} p_b$ is strictly improved by such increase in prices.

Since $k=|\M|+1$ and $p_{\tilde{b}}^T < v_{\tilde{b}}^{(k)} - \epsilon - \delta$, there exists a bidder $n$ with $v_{\tilde{b}}^n\ge v_{\tilde{b}}^{(k)}$ that quits the auction before the auction ends. Let $h$ be the (on-path) history at which $n$ quits (note that at this history, $n$ must not be a leading bidder and $p_{\tilde{b}}(h) \le p_{\tilde{b}}^t\le v_{\tilde{b}}^n-\delta-\epsilon$, where $p_{\tilde{b}}(h)$ is the price of bundle $\tilde{b}$ at the history $h$). Let $h\in I_n$. Then, $s_n(I_n)=\emptyset$. Let $\bm{p}'$ be the raised prices and $\tilde{\bm{b}} \in\arg\max\limits_{\bm{b}\in\mathcal{B}(\M)}\sum_{b'\in\bm{b}} p'_{b'}$. Then $\tilde{b} \in \tilde{\bm{b}}$, and $\bm{p}'$ satisfies the condition in \cref{lem:order}. Therefore, $s_n$ is obviously dominated. We arrive at a contradiction.
\end{proof}

 As we have pointed out following \cref{lem:order}, elimination of obviously dominated strategies does \emph{not} guarantee the ex-post prices to be above the $\k$-highest values. To establish \cref{thm:1}, we prove in addition that the behaviors of non-pivotal bidders are inconsequential. With \cref{thm:1}, we say that CASA is a rank-guaranteed auction format as long as we are comfortable assuming that bidders avoid strategies that are obviously dominated. %When we say another auction format is rank-guaranteed, the solution concept should also be specified. For instance, the SPA is second-guaranteed under dominant strategy equilibrium. The third-priced auction is third-guaranteed under undominated strategies.\footnote{Observe that in the third-priced auction, bidding below value is dominated by bidding the value.}
 
 CASA and its rank-guarantee become extremely simple when $|\M|=1$, where CASA reduces to the canonical English auction of the unique feasible bundle and the $\k$-guarantee becomes the second highest value. In this special case, \cref{thm:1} shares the insight from \cite{li2017obviously} that English auction achieves the $2^{nd}$-guarantee via an obviously strategy-proof outcome. To understand our general result, we provide an alternative interpretation of this special case. The English auction made it obvious that a bidder should avoid losing the auction when there is a remaining surplus. Therefore, it maximally utilizes the competition among \emph{losers} of the auction, pushing the price to the highest value among them. Meanwhile, it does not further screen the \emph{winner} at all. CASA is exactly the multi-item generalization of this philosophy: it maximally invokes competition among losers by making it obviously suboptimal to lose with a remaining surplus. Meanwhile, it does not ``care'' at all which winner gets which bundle and how much extra he pays. As a result, the auctioneer extracts at least as much surplus as that from the losers only, which is exactly the rank-guarantee.
\subsection{Discussions}

\label{sec:discussion}

\paragraph{Alternative formats:} As we have discussed above, CASA is an intuitive extension of the canonical single-item English auction. The idea of generalizing the English auction to accommodate multiple items has been extensively explored. However, crucial to the design of CASA is a set of unique properties that cope with the incentives of the players.  In comparison to the simultaneous ascending auction (SAA, \cite{milgrom2000putting}), allowing for bidding on bundles (package bidding) has the advantage of mitigating the demand reduction problem or the exposure problem, leading to sufficient competition. Compared to the combinatorial variants of SAA like SAAPB (\cite{ausubel2002ascending}) and CCA (\cite{ausubel2006clock}), CASA uses a simpler "pay-as-bid" rule so that the bidders find it straightforward to determine the remaining surplus from the auction, leading to the obvious strategy-proof implementation.\footnote{The SAAPB is indeed $\k$-guaranteed under non-strategic ``straightforward bidding'' strategies (Theorem 1 of \cite{ausubel2002ascending}). However, whether SAAPB is $\k$-guaranteed with fully strategic bidders is yet unknown to us.} Compared to the majority of iterative combinatorial auctions (Chapter 2, \cite{cramton2006combinatorial}), which seek to replicate VCG (known not to be rank-guaranteed\footnote{Imagine the case $\M=\{(a),(b),(a,b)\}$. $v_{(a,b)}=1$ for all bidders. $v_a=v_b=0$ for all bidders except for two, whose value for $a$ and $b$ are $1$. The VCG revenue is $0$, while the $\k$-guarantee is $1$ for any $k\ge 2$.}) and maximize efficiency, CASA has much better revenue performance. The design of CASA accommodates fully strategic bidders, as opposed to single-minded or myopic bidders assumed in the cited papers. 

Except for the several crucial features, the exact format of CASA can be flexibly tailored to the auctioneer's needs and practical considerations without losing its various desirable properties. Notably, the menu $\M$ is a design variable of the auctioneer, which we examine in detail in Section \ref{sec:sufficient}. Additionally, instead of allowing bidders to bid on multiple bundles, we could restrict them to bid on a single bundle when it is their turn to move (if a bidder is already leading for some bundle when it is her turn to move, she could not bid on any other bundle but she would remain an active bidder). Such restriction prevents the type of collusive communication documented in \cite{grimm2003low,jehiel2001european}. Moreover, we could also allow the bidders to observe the entire history if that is considered desirable for transparency purposes.

\paragraph{Robustness to collusion / irrationality:} \cite{klemperer2002really} identifies collusion as the first concern that ``\emph{really matters in auction design}''. We illustrate below that our framework allows us to quantify the impact of collusive or irrational behavior on the performance of CASA. Recall that the key force that drives the revenue to the rank-guarantee is the competition among the losing bidders. Therefore, when there are non-strategic bidders, one can simply consider the competition among the subset of rational losing bidders.

Formally, when the number of non-strategic bidders is bounded, then \cref{thm:1} still holds when $k$ is relaxed by the number of non-strategic bidders.
\begin{prop}
    Suppose there are $j$ non-strategic bidders, then $\underline{R}_{CASA}(\bm{v})\ge R_{\M}^{k}(\bm{v})$ for $k=|\M|+1+j$.
\end{prop}
\begin{proof}
    Observe that in the proof of \cref{thm:1}, since $k=|\M|+1+j$, there exists at least one \emph{strategic} player $n$ that quits the auction before period $t$ and $v_{\tilde{b}}^n\ge v_{\tilde{b}}^{(k)}$. The rest of the proof follows.
\end{proof}
\cref{prop:collusion} applies when a relatively small number of bidders play collusive/irrational strategies. Then, they reduce the rank of the guarantee by a small amount. Nevertheless, CASA is still rank-guaranteed. 

What if \emph{all} bidders are collusive? When the bidders form coalitions, and they strategically maximize group-level payoffs\footnote{Each group can freely shift allocations within the group and maximize the total payoff.}, \cref{thm:1} still holds when $k$ is scaled by the coalition sizes.

\begin{prop}\label{prop:collusion}
    Suppose bidders are partitioned into strategic coalitions $\{c_i\}_{i\in I}$, where the index is chosen such that $|c_i|$ decreases in $i$. Then, $\underline{R}_{CASA}(\bm{v})\ge R_{\M}^{k}(\bm{v})$ for $k=\sum_{i\le |\M|}|c_i|+1$.
\end{prop}
\begin{proof}
    Observe that in the proof of \cref{thm:1}, since $k=\sum_{i\le |\M|}|c_i|+1$, there exists at least one coalition of players $c$ that all quit the auction before period $t$ and $\max_{n\in c}\{v_{\tilde{b}}^n\}\ge v_{\tilde{b}}^{(k)}$. Let $h$ be the (on-path) history at which the last member $n$ in the coalition quits (note that at this history, $n$ must not be a leading bidder and $p_{\tilde{b}}\le p_{\tilde{b}}^t\le \max_n\{v_{\tilde{b}}^n\}-\delta-\epsilon$). Let $h\in I_n$. Then, $s_n(I_n)=\emptyset$. Obviously, quitting gives the entire group zero payoff while bidding $p'_{\tilde{b}}$ guarantees a non-negative payoff.
    Suppose all other bidders bid up to $p_b'$ when it is their turn, the auction ends with $\bm{p}'$ and the group obtains a strictly positive payoff. Therefore, $s_n$ is obviously dominated for coalition $c$.
\end{proof}

The intuition behind the two extensions is exactly the competition among losing strategic bidders. The price of each bundle must be higher than the value of any \emph{losing strategic} bidder or any \emph{losing coalition group} as otherwise they will outbid the price. Of course, \cref{prop:collusion} has no bite when $k$ is large compared to $N$; hence, it should be interpreted as the strategic robustness of CASA only in relatively thick markets. Nevertheless, CASA is also aligned with the philosophy of anti-collusion design, even in thin markets, because CASA permits the minimum transmission of information. For example, anonymity prevents the reciprocity behavior documented in \cite{cramton2000collusive}.

\paragraph{Efficiency:} In CASA, although bidders can fully avoid the exposure problem by bidding only on bundles with a positive surplus, they may strategically expose themselves. It is \emph{not} obviously dominated to bid strictly above the true valuation for a bundle.\footnote{Consider, for example, the case with three bidders $1,2,3$ and three items $a,b,c$ . Bidder $1$ only wants $a$, bidder $2$ only wants $b$, and bidder $3$ only wants the grand bundle $\{a,b,c\}$. The valuation of each bidder for the desired bundle is $1$. By strategically bidding up item $b$ even though bidder $1$ gets zero value from it, bidder $1$ can reduce the bid required for him to win item $a$, creating an exposure problem for $1$.} Therefore, CASA might not satisfy ex-post IR; hence $\k$-guaranteed revenue does not implie $\k$-guaranteed surplus. Of course, CASA still satisfies ex-ante IR (assuming bidders have correct Bayesian priors) since quitting at the beginning is always an option; hence, the ex-ante bounds we derive in the following sections on the revenue of CASA also apply to surplus.

\section{Rank-guarantee as a desideratum}

\label{sec:maxmin}

Section \ref{sec:CASA} shows that CASA achieves the rank-guarantee as long as bidders avoid strategies that are obviously dominated. In this section, we explore the concept of rank-guarantee as a desideratum in auction design.

Clearly, if the auctioneer knows that the bidders' valuations are independent and identically distributed, then rank-guarantee is an appealing approximation when \(N\) is large, as all order statistics converge to the upper bound of the valuation support (at the rate of $\frac{1}{N}$). In what follows, we show that even when the auctioneer has non-Bayesian uncertainty about the joint distribution of the bidders' valuations and maximizes the revenue-guarantee (the worst-case expected revenue where the worst case is taken over all joint distributions that are perceived to be plausible), in many settings, rank-guarantee remains an appealing approximation.

Let $\G \subset \Delta([\underline{v},\overbar{v}]^{2^{S}})$ be an arbitrary subset of distributions of valuation vector. We interpret $\G$ as the auctioneer's estimate of a \emph{representative bidder}'s valuation. Then, the joint distributions of the bidders' valuations that are considered possible by the auctioneer are
\begin{align*}
\F = \Big\{ F \in \Delta([\underline{v}, \overbar{v}]^{N \times 2^{S}}) \, \Big| \, \frac{1}{N} \sum F_n \in \G \Big\},
\end{align*}
where \(F_n\) is the marginal distribution of bidder \(n\)'s valuation. Thus, \(\frac{1}{N}\sum F_n\) is the cumulative distribution function of the valuation of a uniformly randomly selected bidder in the population. We call $\F$ an \emph{ambiguity set}. Such an ambiguity set $\F$ could come from the statistical estimation of $F$ based on a ``sanitized'' dataset about valuations, that is, past bidders' valuations with identity information removed. The ambiguity set $\F$ captures the type of distributional uncertainty introduced by \cite{gc17}, while further generalizing it to capture realistic knowledge structures stemming from statistical inference.\footnote{This setup covers a wide range of scenarios, as $\G$ is completely general. $\G$ could be a singleton set capturing the case in which the auctioneer has no uncertainty about the distribution of a representative bidder's valuation. $\G$ could also be the set of distributions satisfying certain statistical properties (say moment conditions), capturing scenarios in which the auctioneer also has some non-Bayesian uncertainty about the distribution of a representative bidder's valuation.}

%\footnote{\cite{gc17} studies a multi-dimensional screening problem, where the non-Bayesian uncertainty is about the correlation between items. In our setting, we assume the unquantifiable uncertainty is about the correlation between bidders, while among items there may or may not be uncertainty as $\G$ is completely general.}

For any \(F\in \Delta([\underline{v},\overbar{v}]^{N\times 2^{S}}) \) and menu $\M\subseteq 2^S$, define the ex-ante \emph{efficient} surplus
\begin{align*}
V_{\M}(F):=\E_F\left[\max_{\bm{b}\in\mathcal{B}(\M)} \, \max_{\iota:\bm{b}\rightarrow \mathcal{N}} \, \sum_{b\in \bm{b}} \, v^{\iota(b)}_b\right],
\end{align*}
where $\iota$, the assignment function satisfies $\iota(b) \neq \iota(b')$ for any $b \neq b'$. The ex-ante efficient surplus $V_{2^S}(F)$ with respect to the complete menu $2^S$ is denoted by $V^*(F)$ for simplicity.

\begin{thm}\label{thm:maxmin} For any menu $\M$,
\begin{align*}
\inf_{F\in \F} \, \E_F[{R}^k_{\M}(\bm{v})] \ge \inf_{F\in\F} \, V_{\M}(F) - \frac{(k-1)|\M|\overbar{v}}{N}.
\end{align*}
\end{thm}

\begin{proof}

Let \(\n\) be a uniform random element of \(\mathcal{N}\).
\begin{align*}
{R}^k_{\M}(\bm{v}) = 
& \max_{\bm{b} \in \mathcal{B}(\M)} \, \sum_{b\in \bm{b}} \, v^{(k)}_b \\
\ge & \max_{\bm{b} \in \mathcal{B}(\M)} \, \Big( \sum_{b\in \bm{b}} \, v^{\n}_b - \sum_{b \in \bm{b}} \, \max \, \{v^{\n}_b-v^{(k)}_b, 0\} \Big) \\
\ge & \max_{\bm{b} \in \mathcal{B}(\M)} \, \Big( \sum_{b\in \bm{b}} \, v^{\n}_b\Big) - \sum_{b \in \M} \, \max \, \{v^{\n}_b-v^{(k)}_b, 0\} \\
\implies \E_F[{R}^k_{\M}(\bm{v})] \ge & \E_F\left[\max_{\bm{b} \in \mathcal{B}(\M)} \, \sum_{b\in \bm{b}} \, v^{\n}_b\right] - \sum_{b \in \M} \, \E_F\left[ v^{\n}_b-v^{(k)}_b | v^{\n}_b> v^{(k)}_b \right] \text{ Prob}( v^{\n}_b> v^{(k)}_b) \\
\ge & \E_F\left[\max_{\bm{b}\in \mathcal{B}(\M)} \, \sum_{b\in \bm{b}} \, v^{\n}_b\right] - \sum_{b\in \M} \, \overbar{v} \text{ Prob}( v^{\n}_b> v^{(k)}_b)\\
\ge & \E_F\left[\max_{\bm{b}\in \mathcal{B}(\M)} \, \sum_{b\in \bm{b}} \, v^{\n}_b\right] - \frac{(k-1)|\M|\overbar{v}}{N}.
\end{align*}

This further implies that
\begin{align*}
\inf_{F\in\F} \, \E_F\left[ {R}^k_{\M}(\bm{v})\right]\ge& \inf_{F\in \F} \, \E_F\left[\max_{\bm{b}\in \mathcal{B}(\M)} \, \sum_{b\in \bm{b}} \,v^{\n}_b\right]-\frac{(k-1)|\M|\overbar{v}}{N}\\
= & \inf_{G\in\G} \, \E_G\left[\max_{\bm{b}\in \mathcal{B}(\M)} \, \sum_{b\in\bm{b}} \, v_b\right]-\frac{(k-1)|\M|\overbar{v}}{N}\\
\ge & \inf_{F\in\F} \, V_{\M}(F)-\frac{(k-1)|\M| \overbar{v}}{N},
\end{align*}
where the equality follows from the definition of the set $\F$. For the last inequality, observe that for any $G \in \G$, the joint distribution where each bidder's value is distributed according to $G$ and all bidders' values are maximally positively correlated is contained in $\F$. Therefore,
$$\inf_{F\in\F} \, V_{\M}(F)\le \inf_{G\in\G}\E_G\left[\max_{\bm{b}\in \mathcal{B}(\M)} \, \sum_{b\in\bm{b}} \, v_b\right].$$
\end{proof}

%(in which case $V_{\M}(F)\le \E_G\left[\max_{\bm{b}\in \mathcal{B}(\M)} \, \sum_{b\in\bm{b}} \, v_b\right]$)

\cref{thm:maxmin} highlights a key trade-off between \emph{menu sufficiency} and \emph{approximation efficiency}. Evidently, presenting the bidders with a larger menu has the potentially of increasing allocation efficiency. Particularly, \(\M\) can be naively chosen to be the complete menu \(2^S\) to guarantee full menu sufficiency, i.e., $V_{\M}(F)=V^*(F)$. However, this leads to \(|\M|\) growing exponentially in \(M\), causing both complex auction process and slow convergence. On the other hand, choosing a small menu achieves approximation efficiency but sacrifices allocation efficiency. Although such trade-off is generally non-trivial under general combinatorial preferences, we show in the next section that under canonical preference structures, menu sufficiency and approximation efficiency can often be achieved simultaneously. 

\subsection{Discussions}

\paragraph{Alternative ambiguity sets:} The proof of \cref{thm:maxmin} goes through for a general $\F$ if the following equality holds.
\begin{align*}
    \inf_{F\in \F} \, \E_F\left[\max_{\bm{b}\in \mathcal{B}(\M)} \, \sum_{b\in \bm{b}} \,v^{\n}_b\right]= \inf_{F\in\F} \, V_{\M}(F).
\end{align*}
That is, in the worst case, randomly selecting bidders performs as well as optimally selecting bidders. Therefore, it is straightforward that \cref{thm:maxmin} holds under an alternative ambiguity set:
\begin{align*}
    \widehat{\F}=\left\{ F \in \Delta([\underline{v}, \overbar{v}]^{N \times 2^{S}}) \, \Big| \forall n,\  F_n \in \G\right\}.
\end{align*}
Under $\widehat{\F}$, the auctioneer knows that the marginal distributions of bidders are contained in $\G$, but nothing beyond that. This ambiguity set exactly captures the correlational uncertainty studied in \cite{gc17}, \cite{hl21}, \cite{zhang2022correlationrobust}, and \cite{suzdaltsev2022distributionally}.

\paragraph{Distributions with unbounded support:} Suppose the support of distributions in $\G$ has unbounded support, i.e., $\overbar{v}=\infty$, then the bound derived in \cref{thm:maxmin} has no bite. However, observe that
\begin{align*}
    \E_F\left[ v^{\n}_b | v^{\n}_b> v^{(k)}_b \right] \text{ Prob}( v^{\n}_b> v^{(k)}_b)\le \int_{Q_{\bar{F}_b}^{(k-1)/N}}^{\infty} v_b \d \bar{F}(v),
\end{align*}
where $\bar{F} = \frac{\sum F_n}{N}$ and $Q_{\bar{F}_b}^{(k-1)/N}$ is the top $\frac{k-1}{N}$ quantile of the marginal of $\bar{F}$ for $v_b$. To see why the inequality holds, the LHS is the constrained expectation of $v_b$ given $\bar{F}$ on \emph{some} event of probability at most $\frac{k-1}{N}$. The RHS is the constrained expectation of $v_b$ given $\bar{F}$ on the probability-$\frac{k-1}{N}$ event that maximizes its value. Therefore, 
\begin{align*}
    \inf_{F\in \F} \, \E_F[{R}^k_{\M}(\bm{v})] \ge \inf_{F\in\F} \, V_{\M}(F) - \sup_{G\in\G} \sum_{b\in\M}\int_{Q_{G_b}^{(k-1)/N}}^{\infty} v_b \d G(v).
\end{align*}
For any $G$, the concentration inequality implies $\int_{Q_{G_b}^{(k-1)/N}}^{\infty} v_b \d G(v)\le \frac{\E_{G}[v_b^2]}{Q_{G_b}^{(k-1)/N}}$. If $G$ has unbounded support, so is the quantile $Q_{G_b}^{(k-1)/N}$. As a result, when distributions in $\G$ have unbounded support but uniformly bounded second moment, the worst-case rank-guarantee converges to the full surplus at the rate of $O\Big(\frac{|\M|}{\inf_{G\in \G}Q_{G_b}^{(k-1)/N}}\Big)$. 

\paragraph{Tightness of the bound:} The coefficient $k|\M|$ in \cref{thm:maxmin} consists of two parts. The coefficient $k$ comes from the $\k$ highest value approximation. The coefficient $|\M|$ comes from the total number of bundles in the menu $\M$. \cref{prop:tight} below shows that the dependence on $k$ is tight.

\begin{prop}\label{prop:tight} For any $ M,N,\M,k$, there exists some $\G $ such that
\begin{align*}
\inf_{F \in \F} \, \E_F[{R}^k_{\M}(\bm{v})] \le \inf_{F\in\F} \, V_{\M}(F) - \textstyle O \left(\frac{k}{N}\right).
\end{align*}
\end{prop}

\begin{proof}
See \cref{proof:tight}.
\end{proof}

The dependence of the approximation gap on $|\M|$, however, might not be tight as we take a very loose upper bound in the proof of \cref{thm:maxmin}: we bound the revenue loss from the allocated bundles (up to $M$ of them) by the revenue loss from all bundles ($|\M|$ of them). While we speculate that the coefficient $|\M|$ can be improved, a formal proof is yet unknown to us.

\section{Menu design and simple menus}

\label{sec:sufficient}

In this section, we examine several canonical classes of preference structures where there exist menus that are both \emph{sufficient}, ensuring full allocation efficiency, and \emph{simple}, with menu size growing at a polynomial rate as $M$ increases.

\begin{defi}\label{def:sufficiency}
	
Menu $\M$ is $\G$-sufficient if:
\begin{align*}
\inf_{G\in\G} \, \E_G\left[\max_{\bm{b}\in \mathcal{B}(\M)} \, \sum_{b\in\bm{b}} \, v_b\right] = \inf_{G\in\G} \, \E_G\left[\max_{\bm{b}\in \mathcal{B}(2^S)} \, \sum_{b\in\bm{b}} \, v_b\right].
\end{align*}

\end{defi}

In words, a menu $\M$ is $\G$-sufficient if the worst-case surplus from allocating to (hypothetically) identical bidders with valuation distribution from $\G$ is the same as that under the complete menu $2^S$. Importantly, $\G$-sufficiency is defined with respect to the preference of a single bidder instead of all bidders. It is much weaker than assuming that restricting to allocations within $\M$ is without loss for ex-post efficiency.\footnote{Consider for instance two items and two bidders, where for each bidder the sum of the value for each individual item is more than her value for the grand bundle. Then, menu of individual items is ``sufficient'' per \cref{def:sufficiency}, but not necessarily ex-post efficient when the two bidder's values are highly asymmetric.} Nevertheless, sufficiency guarantees full allocation efficiency:

\begin{thm} \label{thm:sufficient}
	
If menu $\M$ is  $\G$-sufficient, then
\begin{align*}
\inf_{F\in\F} \, \E_F\left[ {R}^k_{\M}(\bm{v})\right] \ge \inf_{F\in\F} \, V^*(F) - \frac{(k-1)|\M| \overbar{v}}{N}.
\end{align*}

\end{thm}

\begin{proof}
\begin{align*}
\inf_{F\in\F} \, \E_F\left[ {R}^k_{\M}(\bm{v})\right] \ge& \inf_{G\in\G} \, \E_G\left[\max_{\bm{b}\in \mathcal{B}(\M)} \, \sum_{b\in\bm{b}} \, v_b\right]-\frac{(k-1)|\M|\overbar{v}}{N}\\
= & \inf_{G\in\G} \, \E_G\left[\max_{\bm{b}\in \mathcal{B}(2^S)} \, \sum_{b\in\bm{b}} \, v_b\right] - \frac{(k-1)|\M|\overbar{v}}{N} \\
\ge & \inf_{F\in\F} \, V^*(F)-\frac{(k-1)|\M| \overbar{v}}{N},
\end{align*}
where the equality follows from the $\G$-sufficiency of menu \(\M\), and the two inequalities have been established in the proof of Theorem \ref{thm:maxmin}.
\end{proof}

A simple sufficient condition for the $\G$-sufficiency of menu $\M$ is that
\begin{align}
    \max_{\bm{b} \in \mathcal{B}(2^S)} \, \sum_{b\in\bm{b}} \, v_b = \max_{\bm{b}\in \mathcal{B}(\M)} \, \sum_{b\in\bm{b}} \, v_b \label{eq:sufficient}
\end{align}
 holds ex-post, i.e. $\forall \bm{v}\in \mathrm{Supp}(\G):=\cup_{G\in\G} \, \mathrm{Supp}(G)$. This condition allows us to convert combinatorial preferences into sufficiency. \Cref{thm:sufficient}, as well as \cref{eq:sufficient}, states that one only needs to verify the sufficiency of a menu based on individual bidder's preference, as opposed to the distribution of valuations among all bidders. This is a consequence of the robustness concern. Recall from our analysis in \cref{sec:maxmin} that the adversarial nature minimizes the rank-guarantee by making all losing bidders identical. Then, in the worst-case, a $\G$-sufficient menu performs as good as the complete menu. In the alternative cases where losing bidders are asymmetric, even though the $\G$-sufficient menu under-performs the complete menu, the extra surplus from asymmetry leads to an even higher rank-guarantee. With \cref{thm:sufficient} and \cref{eq:sufficient}, we derive simple sufficient menus for several canonical preference structures.

\paragraph*{Weak substitutability and itemized ascending auction}

\begin{defi}

We say that bidder preferences exhibit \textbf{weak substitutability} if for any \(\bm{v}\in \mathrm{Supp}(\G)\) and \(b \subseteq S\),
\begin{align*}
\sum_{s\in b} v_{\{s\}}\ge v_b.
\end{align*}

\end{defi}

In words, a representative bidder finds the value of any bundle weakly lower than the sum of her value for each item in the bundle. Weak substitutability is a necessary condition for various substitutability notions studied in the literature. 

\begin{prop}
 If bidder preferences exhibit weak substitutability, then the menu \(\M=S\) is \(\G\)-sufficient and \(k=M+1\).
\end{prop}
When \(\M=S\), CASA reduces to a simple itemized ascending auction, where the allocation is determined jointly in the end. Weak substitutability is one of the most widely studied preference assumptions in the literature as it captures a natural diminishing return to scale. Our analysis shows that under such preference structures, CASA exhibits extreme simplicity while achieving both allocation and approximation efficiency. Intriguingly, under weak substitutability, the canonical Vickery auction performs as well as CASA, despite its much worse performance under more general preference structures.

\begin{prop}\label{prop:vcg}
	
 If bidder preferences exhibit weak substitutability, then the Vickery auction achieves a revenue guarantee of ${R}^{M+1}_{S} (\bm{v})$. 
 
\end{prop}
\begin{proof}
    See \cref{proof:vcg}.
\end{proof}
An even more special case of weak substitutability is the sponsored search auction, where valuations of items are constant (and common) ratios of a one-dimensional private type. As shown in \cite{edelman2007internet}, the clock auction version of \emph{generalized second price} (GSP) auction is outcome equivalent to the Vickery auction; hence achieving the same rank-guarantee. 

\paragraph*{Weak complementarity and the second-price auction}

\begin{defi}
  Bidder preferences exhibit \textbf{weak complementarity} if for any \(\bm{v}\in \mathrm{Supp}(\G)\) and \( \bm{b}\in \mathcal{B}(2^S)\),
  \begin{align*}
    \sum_{b\in \bm{b}} v_{b}\le v_S.
  \end{align*}
\end{defi}
In words, a representative bidder finds the value of the grand bundle weakly higher than the total value of any feasible collection of bundles. Weak complementarity is a necessary condition for various complementarity notions studied in the literature. 
\begin{prop}
  If bidder preferences exhibit weak substitutability, then the menu \(\M=\{S\}\) is \(\G\)-sufficient and \(k=2\).
 \end{prop}
 When \(\M=\{S\}\), CASA reduces to a simple ascending auction for only the grand bundle. Evidently, in this case, the standard \emph{second-price auction} is second-guaranteed and outcome-equivalent to CASA.

\paragraph*{``Partitional'' complementarity}

A hybrid case of substitutability and complementarity is the partitional complementarity which we define below, described by a partition $\mathcal{K}$ of $S$. 

\begin{defi}
    Let $\mathcal{K}$ be a partition of $S$. Bidder preferences exhibit $\mathcal{K}$-\textbf{partitioned complementarity} if for any \( \bm{v}\in \mathrm{Supp}(\G)\),
    \begin{align*}
    & \text{for any } b\in \mathcal{K} \text{ and partition $\kappa$ of $b$},\ \sum_{b'\in \kappa} \, v_{b'} \le v_b;\\
    &\text{for any } b'\subseteq S,\ \sum_{b\in \mathcal{K}}v_{b\cap b'}\ge v_{b'}.
\end{align*}
\end{defi}

In words, $\mathcal{K}$-partitioned complementarity structure means there is weak complementarity within each $b\in \mathcal{K}$ and weak substitutability across each $b\in \mathcal{K}$.

\begin{prop} \label{prop:partitional}
If bidder preferences exhibit $\mathcal{K}$-partitioned complementarity, then the menu $\M=\mathcal{K}$ is $\G$-sufficient and $k=|\mathcal{K}|+1$.
\end{prop}

In some cases, the auctioneer may understand that bidder preferences exhibits partitional complementarity, but does not know the exact partition. Proposition \ref{prop:partitional} can be easily extended to the case with multiple possible partitions \(\{\mathcal{K}_i\}_{i=1}^I\), where $I$ is bounded. In this case $\M=\cup _{i\in I}\mathcal{K}_i$ and \(k\sim Poly(M)\). Such partitional complementarity preference structure arises when there is clear synergy between ``nearby'' bundles. Think about land auctions, for example. There are finitely many possible partitions that are determined by the major divisions of lands by rivers, highways, or railroads. If two distinct lands are segregated by those divisions, then there is substitutability among them. In such cases, our theory guarantees the performance of CASA with the partitional menu.

\paragraph*{Homogeneous goods and quantity-CASA}

\begin{defi}
  The goods are \textbf{homogeneous} if there exists \(u:\mathbb{N}\to [\underline{v},\overbar{v}]\) such that for any \( \bm{v}\in \mathrm{Supp}(\G)\) and \( b\in S\),
  \begin{align*}
    v_b=u(|b|).
  \end{align*}
\end{defi}
With homogeneous goods, a representative bidder's valuation for any bundle only depends on the size of the bundle. Note that the dependence of $u$ on $|b|$ is arbitrary. We do not even require monotonicity. In this case, we redefine the notion of feasible allocations to \(\mathcal{B}:(\M)=\{X\subset \M | \sum_{b\in X} \, |b|\le M\}\), i.e., an allocation is feasible as long as the total number of items being allocated is below \(M\). 

\begin{prop}
  If goods are homogeneous, then the menu \(\M=\cup_{l\in\{1,\ldots, M\}} \, \{b_l^1,\ldots, b_l^{\lfloor \frac{M}{l}\rfloor}\}\) is \(\G\)-sufficient and \(k\le \frac{M^2+M}{2}\), where $\{b_l^j\}$ are distinct bundles of size \(l\).
\end{prop}

In this case, CASA simply auctions \(\lfloor \frac{M}{l}\rfloor\) copies of each quantity level \(l\le M\) via individual ascending auctions. Like the discussion in partitional complementarity, there may be finitely many types of homogeneous goods. As long as the number of types $I$ is bounded, the menu consists of all combinations of \(\lfloor \frac{M}{l}\rfloor\) copies of each type is sufficient and of size $Poly(M)$. Such preference structure is typical in examples like the spectrum auctions. Different frequencies are almost physically homogeneous, except that ``middle'' frequencies might be of different value from ``boundary'' frequencies. 

\section{Concluding remarks}

In this paper, we design an auction format of CASA that guarantees an approximately optimal ex-post revenue. To achieve this, we only need to assume minimal rationality on the part of the bidders. In addition, we show that CASA is robust to distributional and strategic uncertainties under certain approximations. In practice, however, these approximation gaps may become non-negligible, rendering the deployment of CASA challenging. 

\begin{itemize}
\item \emph{Thin markets:} The revenue performance of CASA as well as its strategic robustness crucially hinges on the rank $k$ (menu size) being small relative to the number of bidders. In the online advertising examples we introduce, the complete menu is small enough that a handful of bidders may be sufficient to make CASA an appealing design. However, other interesting auctions may suffer the large menu problem (e.g. the land auctions) or the thin market problem (e.g. the route auctions of rideshare apps) or both (e.g., the spectrum auctions), rendering the guarantee underpowered. 

In the latter cases, the menu sufficiency-approximation efficiency tradeoff becomes eminent. Our theory suggests the importance of preference estimation in those settings. Finding a simple sufficient menu keeps the revenue guarantee appealing and CASA directly applicable. Even in settings with a large number of items and a small number of bidders where our theory has little bite, menu design may still be a cost-effective way to promote competition and improve the revenue performance of existing auctions.

\item \emph{Proxy bidding:} While CASA simplifies the bidding process by clarifying ``whether to quit," the complexity of determining ``which bundles to bid on" and ``how much to bid'' remains unresolved. A \emph{truthful} and full proxy-bidding version of CASA is not yet known to us. This makes the deployment of CASA challenging in environments that require fast resolutions of auctions. Nevertheless, we propose that advancements in AI could mitigate this by introducing "copilot" features that assist bidders in decision-making. By integrating AI as the bidding proxy, bidders would only need to specify values for desired bundles, with the AI advising on bid placement. This could evolve into a hybrid model where bidders either rely fully on platform-provided AI, develop their own bidding algorithms, or use a combination of both strategies.
\end{itemize}
\newpage
\bibliographystyle{apalike}
\bibliography{orderstatistics}

\begin{thebibliography}{}

\bibitem[Aggarwal and Hartline, 2006]{aggarwal2006knapsack}
Aggarwal, G. and Hartline, J.~D. (2006).
\newblock Knapsack auctions.
\newblock In {\em SODA}, volume~6, pages 1083--1092.

\bibitem[Ausubel et~al., 2006]{ausubel2006clock}
Ausubel, L.~M., Cramton, P., and Milgrom, P. (2006).
\newblock The clock-proxy auction: A practical combinatorial auction design.
\newblock {\em Handbook of spectrum auction design}, pages 120--140.

\bibitem[Ausubel and Milgrom, 2002]{ausubel2002ascending}
Ausubel, L.~M. and Milgrom, P.~R. (2002).
\newblock Ascending auctions with package bidding.
\newblock {\em The BE Journal of Theoretical Economics}, 1(1):20011001.

\bibitem[Bergemann and Morris, 2005]{bm05}
Bergemann, D. and Morris, S. (2005).
\newblock Robust mechanism design.
\newblock {\em Econometrica}, 73:1771--1813.

\bibitem[B{\"o}rgers, 1991]{borgers1991undominated}
B{\"o}rgers, T. (1991).
\newblock Undominated strategies and coordination in normalform games.
\newblock {\em Social Choice and Welfare}, 8:65--78.

\bibitem[Brooks and Du, 2021]{bd19}
Brooks, B. and Du, S. (2021).
\newblock Optimal auction design with common values: An informationally-robust
  approach.
\newblock {\em Econometrica}, 89(3):1313--1360.

\bibitem[Brooks and Du, 2023]{brooks2023structure}
Brooks, B. and Du, S. (2023).
\newblock On the structure of informationally robust optimal mechanisms.
\newblock {\em Available at SSRN 3663721}.

\bibitem[Bulow and Klemperer, 1996]{bulow1996auctions}
Bulow, J. and Klemperer, P. (1996).
\newblock Auctions versus negotiations.
\newblock {\em American Economic Review}, 86(1):180--194.

\bibitem[Carroll, 2014]{carroll2014complexity}
Carroll, G. (2014).
\newblock A complexity result for undominated-strategy implementation.
\newblock Technical report, Working Paper.

\bibitem[Carroll, 2017]{gc17}
Carroll, G. (2017).
\newblock Robustness and separation in multidimensional screening.
\newblock {\em Econometrica}, 85(2):453--488.

\bibitem[Carroll, 2019]{carroll2019design}
Carroll, G. (2019).
\newblock Design for weakly structured environments.
\newblock {\em The Future of Economic Design: The Continuing Development of a
  Field as Envisioned by Its Researchers}, pages 27--33.

\bibitem[Chen and Li, 2018]{cl18}
Chen, Y.-C. and Li, J. (2018).
\newblock Revisiting the foundations of dominant-strategy mechanisms.
\newblock {\em Journal of Economic Theory}, 178:294--317.

\bibitem[Chung and Ely, 2007]{ce07}
Chung, K.-S. and Ely, J.~C. (2007).
\newblock Foundations of dominant-strategy mechanisms.
\newblock {\em Review of Economic Studies}, 74(2):447--476.

\bibitem[Cramton and Schwartz, 2000]{cramton2000collusive}
Cramton, P. and Schwartz, J.~A. (2000).
\newblock Collusive bidding: Lessons from the fcc spectrum auctions.
\newblock {\em Journal of regulatory Economics}, 17(3):229--252.

\bibitem[Cramton et~al., 2006]{cramton2006combinatorial}
Cramton, P.~C., Shoham, Y., Steinberg, R., and Smith, V.~L. (2006).
\newblock {\em Combinatorial auctions}, volume~1.
\newblock MIT press Cambridge.

\bibitem[Du, 2018]{du18}
Du, S. (2018).
\newblock Robust mechanisms under common valuation.
\newblock {\em Econometrica}, 86(5):1569--1588.

\bibitem[Edelman et~al., 2007]{edelman2007internet}
Edelman, B., Ostrovsky, M., and Schwarz, M. (2007).
\newblock Internet advertising and the generalized second-price auction:
  Selling billions of dollars worth of keywords.
\newblock {\em American economic review}, 97(1):242--259.

\bibitem[Goldberg and Hartline, 2001]{goldberg2001competitive}
Goldberg, A.~V. and Hartline, J.~D. (2001).
\newblock Competitive auctions for multiple digital goods.
\newblock In {\em European Symposium on Algorithms}, pages 416--427. Springer.

\bibitem[Grimm et~al., 2003]{grimm2003low}
Grimm, V., Riedel, F., and Wolfstetter, E. (2003).
\newblock Low price equilibrium in multi-unit auctions: the gsm spectrum
  auction in germany.
\newblock {\em International journal of industrial organization},
  21(10):1557--1569.

\bibitem[Hartline, 2013]{hartline2013mechanism}
Hartline, J.~D. (2013).
\newblock Mechanism design and approximation.
\newblock {\em Book draft. October}, 122(1).

\bibitem[He and Li, 2022]{hl21}
He, W. and Li, J. (2022).
\newblock Correlation-robust auction design.
\newblock {\em Journal of Economic Theory}, 200:105403.

\bibitem[He et~al., 2022]{he2022order}
He, W., Li, J., and Zhong, W. (2022).
\newblock Order statistics of large samples: theory and an application to
  robust auction design.
\newblock Technical report, Technical report, Mimeo.

\bibitem[Jackson, 1992]{jackson1992implementation}
Jackson, M.~O. (1992).
\newblock Implementation in undominated strategies: A look at bounded
  mechanisms.
\newblock {\em The Review of Economic Studies}, 59(4):757--775.

\bibitem[Jehiel and Moldovanu, 2001a]{jehiel2001efficient}
Jehiel, P. and Moldovanu, B. (2001a).
\newblock Efficient design with interdependent valuations.
\newblock {\em Econometrica}, 69(5):1237--1259.

\bibitem[Jehiel and Moldovanu, 2001b]{jehiel2001european}
Jehiel, P. and Moldovanu, B. (2001b).
\newblock The european umts/imt-2000 licence auctions.
\newblock {\em Imt-2000 Licence Auctions (May 2001)}.

\bibitem[Klemperer, 2002]{klemperer2002really}
Klemperer, P. (2002).
\newblock What really matters in auction design.
\newblock {\em Journal of economic perspectives}, 16(1):169--189.

\bibitem[Levin and Skrzypacz, 2016]{levin2016properties}
Levin, J. and Skrzypacz, A. (2016).
\newblock Properties of the combinatorial clock auction.
\newblock {\em American Economic Review}, 106(9):2528--2551.

\bibitem[Li and Dworczak, 2021]{li2021simple}
Li, J. and Dworczak, P. (2021).
\newblock Are simple mechanisms optimal when agents are unsophisticated?
\newblock In {\em Proceedings of the 22nd ACM Conference on Economics and
  Computation}, pages 685--686.

\bibitem[Li, 2017]{li2017obviously}
Li, S. (2017).
\newblock Obviously strategy-proof mechanisms.
\newblock {\em American Economic Review}, 107(11):3257--3287.

\bibitem[Milgrom, 2000]{milgrom2000putting}
Milgrom, P. (2000).
\newblock Putting auction theory to work: The simultaneous ascending auction.
\newblock {\em Journal of political economy}, 108(2):245--272.

\bibitem[Myerson, 1981]{my81}
Myerson, R. (1981).
\newblock Optimal auction design.
\newblock {\em Mathematics of Operations Research}, 6(1):58--71.

\bibitem[Rothkopf et~al., 1998]{rothkopf1998computationally}
Rothkopf, M.~H., Peke{\v{c}}, A., and Harstad, R.~M. (1998).
\newblock Computationally manageable combinational auctions.
\newblock {\em Management science}, 44(8):1131--1147.

\bibitem[Roughgarden, 2015]{roughgarden2015approximately}
Roughgarden, T. (2015).
\newblock Approximately optimal mechanism design: Motivation, examples, and
  lessons learned.
\newblock {\em ACM SIGecom Exchanges}, 13(2):4--20.

\bibitem[Suzdaltsev, 2022]{suzdaltsev2022distributionally}
Suzdaltsev, A. (2022).
\newblock Distributionally robust pricing in independent private value
  auctions.
\newblock {\em Journal of Economic Theory}, 206:105555.

\bibitem[Yamashita, 2015]{yamashita2015implementation}
Yamashita, T. (2015).
\newblock Implementation in weakly undominated strategies, with applications to
  auctions and bilateral trade.
\newblock {\em Review of Economic Studies}, 82(3):1223--1246.

\bibitem[Yamashita and Zhu, 2022]{yamashita2022foundations}
Yamashita, T. and Zhu, S. (2022).
\newblock On the foundations of ex post incentive-compatible mechanisms.
\newblock {\em American Economic Journal: Microeconomics}, 14(4):494--514.

\bibitem[Zhang, 2022]{zhang2022correlationrobust}
Zhang, W. (2022).
\newblock Correlation-robust optimal auctions.

\end{thebibliography}

%\nocite{*}
\newpage
\appendix

\section{Omitted Proofs}

\subsection{Proof of \cref{prop:tight}}\label{proof:tight}

\begin{proof}

Pick an arbitrary bundle $b\in \M$. Let $v_{b'}=\bm{1}_{b'=b}\cdot U[0,1]$; that is, $b$ is the only valuable bundle and its value is uniformly distributed on $[0,1]$. Let $G$ denote such a distribution and $\G=\{G\}$. Then, $V_{\M}(F)\ge \frac{1}{2}$ for any $F\in \F$. Define $F^*$ as follows: uniformly randomly pick $k-1$ bidders and their values for $b$ are identical and distributed according to $U[1-\frac{k-1}{N},1]$. For the remaining bidders, their values for $b$ are identical and distributed according to $U[0,1-\frac{k-1}{N}]$. It is straightforward to verify that $F^*\in \F$ and 
\begin{align*}
\E_{F^*}[R^k_{\M}(\bm{v})]=\E_{U[0,1-\frac{k-1}{N}]}[x]=\frac{1}{2}-\frac{k-1}{2N}\le \inf_{F\in\F} \, V_{\M}(F)-\textstyle O\left(\frac{k}{N}\right).
\end{align*}
\end{proof}

\subsection{Proof of \cref{prop:vcg}}\label{proof:vcg}

\begin{proof}

We slightly abuse notation and represent an allocation by a vector of sets $\bm{b}=(b_1,b_2, \ldots, b_N)$, where $b_n\cap b_{n'}=\emptyset$ and $b_n$ is the bundle allocated to bidder $n$. Let $\mathcal{B}_N$ denote the set of all feasible allocations with $N$ bidders. Let $\bm{b}^*(\bm{v})$ denote the efficient allocation.

We establish a lower bound of the revenue-guarantee of the VCG mechanism by constructing, for each $n$, an allocation $\bm{b}^n \in \mathcal{B}_{N - 1}$ of the objects to the bidders other than bidder $n$. Clearly, for any such profile $\bm{b}^n$,
\begin{align}
R_{VCG} (\bm{v}) = & \sum_{n=1}^N \left( \sup_{\bm{b} \in \mathcal{B}_{N-1}} \sum_{n' \neq n} v^{n'}_{b_{n'}} -\sum_{n' \neq n} v^{n'}_{b^*_{n'} (\bm{v})} \right) \notag \\
\ge & \sum_{n=1}^N \left( \sum_{n' \neq n} v^{n'}_{b^n_{n'}} -\sum_{n' \neq n} v^{n'}_{b^*_{n'} (\bm{v})} \right). \label{eqn:2}
\end{align}

For each $n$, we construct an allocation $\bm{b}^n \in \mathcal{B}_{N - 1}$ via the following algorithm: \medskip

\textsl{Algorithm.} Bundle $b^n_{n'} = \emptyset$ for all $n'$. Set $O = b^*_n(\bm{v})$.
\begin{itemize}
	\item[(1).] For each $n' \neq n$:
	
	\qquad If $b^*_{n'} (\bm{v})\neq \emptyset$, set $b^n_{n'} = b^*_{n'}(\bm{v})$.
	
	\qquad Let $\bar{N} = \{n': b^n_{n'} = \emptyset, n' \neq n\}$.
	
	\item[(2).] If $O \neq \emptyset$, then pick $o \in O$.
	
	\qquad Set $b^n_{n'} = \{o\}$ for some $n' \in \arg\max_{n'' \in \bar{N}} v_{n''} (\{o\})$.
	
	\qquad Update $O \leftarrow O \setminus \{o\}$ and $\bar{N} \leftarrow \bar{N} \setminus \{n'\}$.
	
	\item[(3).] Repeat (2) until $O = \emptyset$.
	
    \item[(4).] Return allocation $\bm{b}^n = (b^n_1, b^n_2, \ldots, b^n_{n - 1}, b^n_{n + 1}, \ldots, b^n_N)$.
\end{itemize}

 In words, if an object is allocated to a bidder other than bidder $n$ under $\bm{b}^*(\bm{v})$, then the object is still allocated to that bidder. We then iteratively pick an object $o$ that is allocated to bidder $n$ under $\bm{b}^*(\bm{v})$, and allocate the object to the bidder $n'$ whose value for the object $v^{n'}_{\{o\}}$ is the highest among all the bidders who are not allocated any object yet.  For each $o \in b^*_n$, define $n_o$ to be the index $n'$ such that $b^n_{n'} = \{o\}$. \par
 
 It follows from \cref{eqn:2} that
\begin{align}
R_{VCG}(\bm{v})\ge & \sum_{n=1}^N \left( \sum_{n' \neq n} v^{n'}_{b^n_{n'}} -\sum_{n' \neq n} v^{n'}_{b^*_{n'}(\bm{v})} \right) \notag \\
= & \sum_{n=1}^N \sum_{o \in b^*_n} v^{n_o}_{\{o\}} \notag \\
\ge & \sum_{o = 1}^M v^{(M+1)}_{\{o\}} \\ 
=& {R}^{M+1}_{\M}(\bm{v}). \notag
\end{align}
The first equality holds since (a) when $b^*_{n'} (\bm{v})\neq \emptyset$, $b^n_{n'} = b^*_{n'}(\bm{v})$, and (b) when $b^*_{n'} (\bm{v}) = \emptyset$, $b^n_{n'}$ is either $\{o\}$ for some $o \in O$, or $\emptyset$ otherwise. The second inequality follows from the construction of $\bm{b}^i$: when an object $o \in b^*_i$ is being allocated, it is allocated to the bidder $n'$ whose value for the object $v^{n'}_{\{o\}}$ is the highest among all the bidders who are not allocated any object yet. Since each iteration assigns at least one good to one bidder and there are at most $M$ goods, we have $v^{n'_o}_{\{o\}}$ must be at least the $(M + 1)^{th}$ highest value among all $v^n_{\{o\}}$.
\end{proof}
\end{document}